\documentclass[12pt,a4paper]{amsart} 
\usepackage{mathrsfs}
\usepackage{mathtools}
 \usepackage{textcomp}
 
\theoremstyle{plain}                 
\newtheorem{theorem}{Theorem}[section]

\newtheorem{lemma}[theorem]{Lemma}       
\theoremstyle{definition}

\theoremstyle{remark}       
\newtheorem{remark}{Remark}

\DeclarePairedDelimiter\floor{\lfloor}{\rfloor}

\newcommand{\sta}{\mathcal}

\newcommand{\MMMbar}{\overline{\sta M}}
\def\Z{\mathbb{Z}}
\def\ch{\mathrm{ch}}
\def\C{\mathrm{C}}
\def\bbC{\mathbb{C}}
\def\oM{\overline{\mathcal{M}}}
\def\cW{{\mathcal{W}}}

\def\J{\mathrm{J}}
\def\I{\sqrt{-1}}

\def\dd{\mathrm{d}}
\def\CP1{\mathbb{C}\mathrm{P}^1}
\def\cB{\mathcal{B}}
\def\cO{\mathcal{O}}
\def\ident{
\begin{picture}(18,10)
\put(3,-1){$\rightarrow$}
\put(3.5,2.5){$\sim$}
\end{picture}
}
\def\cC{\mathcal{C}}
\def\oC{\overline{\mathcal{C}}}

\begin{document}
\title[Chiodo formulas and topological recursion]{Chiodo formulas for the $r$-th roots and topological recursion}
\author{D.~Lewanski}

\address{D.L.: Korteweg-de Vries Institute for Mathematics, University of Amsterdam, Postbus 94248, 1090 GE Amsterdam, The Netherlands}
\email{D.Lewanski@uva.nl}

\author{A.~Popolitov}

\address{A.P.: Korteweg-de Vries Institute for Mathematics, University of Amsterdam, Postbus 94248, 1090 GE Amsterdam, The Netherlands and ITEP, Moscow, Russia}
\email{A.Popolitov@uva.nl}

\author{S.~Shadrin}

\address{S.S.: Korteweg-de Vries Institute for Mathematics, University of Amsterdam, Postbus 94248, 1090 GE Amsterdam, The Netherlands}
\email{S.Shadrin@uva.nl}

\author{D.~Zvonkine}

\address{D.Z.: Institut de MathÃ©matiques de Jussieu - Paris Rive Gauche and CNRS 4, place Jussieu Case 247, 75252 PARIS Cedex 05}

\email{zvonkine@math.jussieu.fr}

\begin{abstract}
We analyze Chiodo's formulas for the Chern classes related to the $r$-th roots of the suitably twisted integer powers of the canonical class on the moduli space of curves. The intersection numbers of these classes with $\psi$-classes are reproduced via the Chekhov-Eynard-Orantin topological recursion. 
	
	As an application, we prove that the Johnson-Pandharipande-Tseng formula for the orbifold Hurwitz numbers is equivalent to the topological recursion for the orbifold Hurwitz numbers. In particular, this gives a new proof of the topological recursion for the orbifold Hurwitz numbers.
\end{abstract}

\maketitle

\tableofcontents
\section{Introduction}

\subsection{Topological recursion} 
The topological recursion in the sense of Chekhov, Eynard, and Orantin (see, e.g.,~\cite{EO}) takes as an input 
a spectral curve $(\Sigma, x, y, B )$, i.e.,~the data of a Riemann surface $\Sigma$, two functions $x$ and $y$ on $\Sigma$ with some compatibility condition, and the choice of a bi-differential $B$ on the surface (which is canonical in the case $\Sigma = \CP1$, so we will omit it in this case). The recursion produces a collection of symmetric $n$-differentials $\mathcal{W}_{g,n}$ (called correlation differentials) defined again on the surface
 whose expansion can generate solutions to enumerative geometric problems.


In particular, under some conditions the expansion of $\mathcal{W}_{g,n}$ are related to the correlators of semi-simple  cohomological field theories~\cite{DOSS}. 

\subsection{Chiodo's formula}
In \cite{Mu_GRR}, Mumford derived a formula for the Chern character of the Hodge bundle on the moduli space of curves $\overline{\mathcal{M}}_{g,n}$ in terms of the tautological classes and Bernoulli numbers. 
In \cite{Chiodo}, Chiodo generalizes Mumford's formula. The moduli stack $\MMMbar_{g,n}$ is substituted with $\MMMbar^{r,s}_{g; a_1, \dots, a_n}$, the proper moduli stack of
$r$th roots of the line bundle 
$$\omega_{\rm log}^{\otimes s}\left(-\sum_{i=1}^n a_ix_i\right)$$
where $\omega_{\rm log} = \omega(\sum x_i)$, the integers $s$, $a_1, \dots, a_n$ satisfy $$(2g-2+n)s-\sum_i a_i\in r\Z,$$ and the $x_i$'s are the marked points on the curves. Let $\pi\colon \mathcal{C}\to\MMMbar^r_{g,n}$ be the universal curve and denote by $\mathcal{S} \to \mathcal{C}$ the universal $r$-th root. 
Chiodo's formula computes the Chern character $\mathrm{ch}(R^{\bullet} \pi_{*}\mathcal S)$, again in terms of tautological classes and values of Bernoulli polynomials at rational points with denominator~$r$. The push-forward of the corresponding Chern class to the moduli space of curves will be called the {\em Chiodo class}.

In one particular case we know a relation between the Chiodo classes and the topological recursion. Namely, the coefficients of some expansion of the differentials $\mathcal{W}_{g,n}$ for the spectral curve data $(\Sigma=\CP1, x=\log z-z^r, y=z)$ are expressed in terms of the intersection numbers of the Chiodo classes for $s=1$, $r=1,2,\dots$. The main result of this paper is an extension of this correspondence to arbitrary $s\geq 0$. 

\subsection{Chiodo classes and topological recursion}

We consider the spectral curve 
\begin{equation}\label{eq:intro-spectral-curve}
(\Sigma=\CP1,
x(z)=-z^r+\log z, 
y(z)=z^s).
\end{equation}
 We prove that (see Theorem \ref{TH1})
 
 \medskip
 \noindent
 \emph{the expansion of the corresponding correlation differentials in some auxiliary basis of $1$-forms is given by the intersection numbers of the corresponding Chiodo class for these particular $r,s\geq 1$. }

\medskip

The case $s=0$ is exceptional. In this case, the intersection numbers are the same as in the case $s=r$, so we still have to use the spectral curve $(\Sigma=\CP1,
x(z)=-z^r+\log z, 
y(z)=z^r)$.

These spectral curves are known in the literature, in some particular cases, in relation to various versions of Hurwitz numbers.

\subsection{Hurwitz numbers}
Hurwitz numbers play an important role in the interaction of combinatorics, representation theory of symmetric groups, integrable systems, tropical geometry, matrix models, and intersection theory on moduli spaces of curves.

There are several kinds of Hurwitz numbers. Simple Hurwitz numbers enumerate finite degree $d$ coverings of the 2-sphere by a genus $g$ connected surface, with a fixed ramification profile $(\mu_1,\dots,\mu_n)$ over infinity, $\sum_{i=1}^n \mu_i =d$ while the remaining $2g - 2 + n+d$ ramifications over fixed points are simple.

These Hurwitz numbers are known to be the coefficients of the expansions of the correlation forms of the spectral curve~\eqref{eq:intro-spectral-curve} for $r=s=1$. This was conjectured in~\cite{BM} and proved in several different ways, see, e.g.,~\cite{EMS09,DKOSS}.

Chiodo's formula in this case is reduced to the standard Mumford formula, so the Chiodo class is the Chern class of the dual Hodge bundle on the moduli space of curves. The fact that the same correlation differentials are related, in different expansion, to simple Hurwitz numbers and to the intersection numbers, implies that there is a formula for simple Hurwitz numbers in terms of the intersection numbers. Indeed, it is the celebrated ELSV formula~\cite{ELSV}. The equivalence between the topological recursion and the ELSV formula is proved in~\cite{Eyn11}, see also~\cite{DKOSS,SSZ}.

Another example is $r$-spin Hurwitz numbers. In this case, the definition is a bit involved; roughly speaking, we still consider the maps of genus $g$ algebraic curves to $\CP1$, with a fixed profile over infinity, but the remaining simple ramifications are replaced by more complicated singularities, so-called completed cycles. We refer to~\cite{SSZ2,SSZ} for the precise definition.

In this case, the $r$-spin Hurwitz numbers are conjecturally related by the spectral curve~\eqref{eq:intro-spectral-curve} for that particular $r$ and $s=1$, see~\cite{MSS,SSZ}. The same logic as for the simple Hurwitz numbers implies that this conjecture is equivalent to an ELSV-type formula that expresses the $r$-spin Hurwitz numbers in terms of intersection numbers~\cite{SSZ}. The corresponding ELSV-type formula was conjectured in~\cite{Z} and is still open.

\subsection{Orbifold Hurwitz numbers}

A case of special interest for us is the $r$-orbifold Hurwitz numbers. They enumerate finite degree $d$, $r|d$, coverings of the 2-sphere by a genus $g$ connected surface, with a fixed ramification profile $(\mu_1,\dots,\mu_n)$ over the infinity, $\sum_{i=1}^n \mu_i =d$, the fixed ramification profile $(r,r,\dots,r)$ over zero, while the remaining $2g - 2 + n+d/r$ ramifications over fixed points are simple.

It is proved in~\cite{BHLM,DLN} that the $r$-orbifold Hurwitz numbers satisfy the topological recursion for the spectral curve~\eqref{eq:intro-spectral-curve} with this particular $r$ and $s=r$. Johnson-Pandharipande-Tseng~\cite{JoPaTs} exhibited an ELSV-type formula that can be restricted to express $r$-orbifold Hurwitz numbers in terms of intersection numbers. As an application of the general correspondence between the Chiodo formulas and topological recursion, we prove the equivalence of these two statements (see Theorem 5.1). 

Since the Johnson-Pandharipande-Tseng formula (the JPT formula, for brevity) is proved independently, our equivalence result implies a proof of the topological recursion of $r$-orbifold Hurwitz numbers. 

It is a new proof of the topological recursion; the existing proofs~\cite{BHLM,DLN} do use the JPT formula, but only its combinatorial structure, and not the geometry of the classes. The topological recursion is then derived in~\cite{BHLM,DLN} from an additional recursion relation for $r$-orbifold Hurwitz numbers called cut-and-join equation. 

\subsection{Further remarks}

A natural question is whether we can use the equivalence between the topological recursion and the JPT formula for $r$-orbifold Hurwitz numbers in order to give a new proof of the JPT formula, as it is done in~\cite{DKOSS} for the simple Hurwitz numbers. 
This approach requires a new proof of the topological recursion that wouldn't use the JPT formula. This is done in~\cite{DLPS}, so we refer there for further details.

Another natural question is whether there is any natural combinatorial and/or geometric problem of Hurwitz type related to the other spectral curves~\eqref{eq:intro-spectral-curve} for arbitrary $r$ and $s$. The only indication of a possible relation that we know is that similar spectral curves are used in~\cite{MSS} for the so-called mixed Hurwitz numbers in the context of the quantum spectral curve theory.

\subsection{Plan of the paper}

In Section 2 we review the semi-simple cohomological field theories, possibly with a non-flat unit, that correspond to  Chiodo classes. 
In Section 3 we recall the general formula of the differentials $\mathcal{W}_{g,n}$ in terms of integrals over moduli spaces of curves as described in \cite{DOSS, E}, while in Section 4 we compute explicitly all the ingredients of that formula and prove our main theorem, Theorem~\ref{TH1}. Finally, in Section 5 we identify the particular Chiodo class with the one used in the JPT formula and prove the equivalence of the JPT formula and the topological recursion for $r$-orbifold Hurwitz numbers.

\subsection{Acknowledgments}
We are grateful to P.~Dunin-Barkowski for useful discussions and to Paolo Rossi for pointing out an incorrect formula in the first version of the paper. D.L., A.P., and S.S. were supported by the Netherlands Organization for Scientific Research. The authors were also supported by the van Gogh program of the French-Dutch Academy. A.~P. was also partially supported by the Russian President's Grant of Support for the Scientific Schools NSh-3349.2012.2 and by RFBR grants 13-02-00457 and 14-01-31492-mol\_a.

\section{Chiodo classes}
In this Section we recall the definition and some simple properties of the Chiodo classes. These classes are defined on the moduli spaces of tensor $r$th roots of the line bundle $\omega_{\rm log}^{\otimes s} \left( -\sum m_i x_i \right)$, but in this paper we will only need their push-forward to the space of curves $\oM_{g,n}$. A more detailed discussion of the space of $r$th roots in the case $s=0$ is contained in Section~\ref{sec:geometry}. We also refer the reader to~\cite{Chiodo,ChiZvo,ChiRua,SSZ} for all necessary background and origin of the lemmas in this section.

\subsection{Definition}
Let $r\geq 1$ be an integer and $1\leq a_1,\dots,a_n\leq r$, $0\leq s$ be integers satisfying 
\begin{equation} \label{Eq:modrcondition}
(2g-2+n)s-\sum_{i=1}^n a_i \in r\Z
\end{equation}

Consider the morphisms
$$
\oC \stackrel{\pi}\to \oM^{r,s}_{g;a_1, \dots, a_n} \stackrel{\epsilon}\to \oM_{g,n},
$$
where $\oM^{r,s}_{g;a_1, \dots, a_n}$ is the space of $r$th roots ${\mathcal S}^{\otimes r} \ident  \omega_{\rm log}^{\otimes s} \left( -\sum a_i x_i \right)$, $\oC$ is its universal curve, and $\epsilon$ is the forgetful morphism to the space of curves.

While the boundary strata of $\oM_{g,n}$ are described by stable graphs, those of $\oM^{r,s}_{g;a_1, \dots, a_n}$ are described by stable graphs with a remainder mod $r$ assigned to each half-edge in such a way that the sum of residues on each edge vanishes and that Condition~\eqref{Eq:modrcondition} is satisfied for each vertex. The boundary divisors correspond to one-edged graphs with two opposite remainders mod~$r$ assigned the two half-edges.

The Chern characters of the derived push-forward $R^* \pi_* {\mathcal S}$ are given by Chiodo's formula~\cite{Chiodo}
\begin{align}
\ch_m(R^* \pi_* {\mathcal S})  =		
		\frac{B_{m+1}(\frac sr)}{(m+1)!} \kappa_m
		- \sum_{i=1}^n 
		\frac{B_{m+1}(\frac{a_i}r)}{(m+1)!} \psi_i^m 
		\\ \notag 
		+ \frac{r}2 \sum_{a=0}^{r-1} 
		\frac{B_{m+1}(\frac{a}r)}{(m+1)!} (j_a)* 
		\frac{(\psi')^m + (-1)^{m-1} (\psi'')^m}{\psi'+\psi''},
\end{align}
where $j_a$ is the boundary map corresponding to the boundary divisor with remainder~$a$ at one of the two half-edges and $\psi',\psi''$ are the $\psi$-classes at the two branches of the node.

We are interested in the Chiodo classes 
\begin{align}\label{eq:DefinitionChiodoCohFT}
& \C_{g,n}(r,s;a_1,\dots,a_n) = 
\\ \notag 
& \epsilon_* c(-R^*\pi_* {\mathcal S}) = \\ \notag
&\epsilon_* \left[ c(R^1\pi_*{\mathcal S})/c(R^0\pi_*{\mathcal S})\right] =
 \\ \notag
& \epsilon_{*}\exp\left(\sum_{m=1}^\infty (-1)^m (m-1)!\ch_m(R^*\pi_* {\mathcal S})\right) \in H^{\rm even}(\oM_{g,n}).
\end{align}
An explicit expression of the classes $\C_{g,n}(r,s;a_1,\dots,a_n)$ in terms of stable graphs, obtained by expanding the exponential in the expression above, is given in \cite{JPPZ}, Corollary~4.

Consider $\C_{g,n}(r,s;a_1,\dots,a_n)$ as a coefficient of a map 
\begin{equation}
\C_{g,n}(r,s)\colon V^{\otimes n}\to H^{\rm even}(\oM_{g,n}),
\end{equation}
where $V=\langle v_1,\dots,v_r\rangle$, and 
\begin{equation}
\C_{g,n}(r,s)\colon v_{a_1}\otimes\cdots\otimes v_{a_n} \mapsto
\C_{g,n}(r,s;a_1,\dots,a_n).
\end{equation}

\subsection{Cohomological field theories}

\begin{lemma}\label{lem:Chiodo-CohFT} For $0\leq s\leq r$ the classes $\{\C_{g,n}(r,s)\}$ form a semi-simple cohomological field theory.
\end{lemma}

A semi-simple cohomological field theory (CohFT) is obtained via the action of an element of the  upper-triangular Givental group on a topological field theory. In order to determine a topological field theory $\{\omega_{g,n}\}$, we have to fix its scalar product $\eta$ and $\omega_{0,3}$. An element of the upper-triangular Givental group is determined by a matrix $R(\zeta) \in {\rm End}(V) [[ \zeta]]$ that should satisfy the symplectic conditions with respect to $\eta$. 

In the case of  $\{\C_{g,n}(r,s)\}$ we have the following description.
\begin{lemma}\label{lem:givental-cohft} For $0\leq s\leq r$ the classes  $\{\C_{g,n}(r,s)\}$ are given by Givental's action of the $R$-matrix $R(\zeta)$ on the topological field theory $\omega$ with metric $\eta$ on~$V$, where  
\begin{align}
 & V = \langle v_1, \dots, v_r \rangle, \\
	& R(\zeta) = \exp\left(\sum_{m=1}^{\infty}  \frac{\mathrm{diag}_{a=0}^{r-1} B_{m+1}\left(\frac ar\right)}{m(m+1)}(-\zeta)^m\right), \\
	& R^{-1}(\zeta) = \exp\left(-\sum_{m=1}^{\infty}  \frac{\mathrm{diag}_{a=0}^{r-1} B_{m+1}\left(\frac ar\right)}{m(m+1)}(-\zeta)^m\right), \\
	& \eta(v_a,v_b)= \frac{1}{r}\delta_{a+b \mod r}, \\ 
	& \omega_{0,3}(v_a\otimes v_b\otimes v_c) =\frac 1r \delta_{a+b+c-s \mod r},\\
	& \omega_{g,n}(v_{a_1}\otimes \cdots \otimes v_{a_n}) =r^{2g-1} \delta_{a_1+\cdots+a_n-s(2g-2+n) \mod r}.
\end{align}
\end{lemma}

\subsection{Cohomological field theories with a non-flat unit}\label{sec:CohFT/1}

Let us discuss now what happens for $s>r$. We need an extension of the notion of cohomological field theory, namely, we have to consider the cohomological field theories with a non-flat unit, CohFT/1 for brevity. 

The CohFT/1s are obtained by an extension of the Givental group by translations, which allows one to use the dilaton leaves (in the terminology of~\cite{DSS,DOSS}) or $\kappa$-legs (in the terminology of~\cite{PPZ}) with arbitrary coefficients. We refer to the exposition in~\cite{PPZ} for further details.

One of the possible descriptions of a CohFT/1 is in terms of stable graphs without any $\kappa$-legs. The vertices, leaves, and edges of these graphs are decorated in exactly the same way as in the case of a usual CohFT, but in addition every vertex is also decorated by $\exp(\sum_{m=1}^\infty T_m\kappa_m)$ for some constants $T_m$, $m=1,2,\dots$.

In the case of Chiodo classes~\eqref{eq:DefinitionChiodoCohFT} for $s>r$, we have the following:

\begin{lemma}\label{lem:givental-CohFT-1} For $s > r$ the classes $\{\C_{g,n}(r,s)\}$ form a CohFT/1. The corresponding element of the extended Givental group coincides with the one described in Lemma~\ref{lem:givental-cohft}, but instead of the dilaton shift, we decorate each vertex by 
\begin{equation}
\exp\left(\sum_{m=1}^\infty (-1)^m \frac{B_{m+1}(\frac sr)}{m(m+1)} \kappa_m \right).
\end{equation}
\end{lemma}


\section{Topological recursion and Givental group}

In this Section we revisit the main result of~\cite{DOSS,E}. We present a version a bit refined of it, in order to make precise relation that incorporates a torus action on cohomological field theories.

\subsection{General background}

The input of the local topological recursion consists of a local spectral curve $\Sigma=\sqcup_{i=1}^r U_i$, which is a disjoint union of open disks with the center points $p_i$, $i=1,\dots,r$, holomorphic function $x\colon \Sigma\to\bbC$ such that the zeros of its differential $dx$ are $p_1,\dots,p_r$, holomorphic function $y\colon \Sigma\to\bbC$, and a symmetric bidifferential $B$ defined on $\Sigma\times \Sigma$ with a double pole on the diagonal with residue $1$. 

The output is a set of symmetric differentials $\cW_{g,n}$ on $\Sigma^n$. This set of differentials is canonically associated to the input data via the topological recursion procedure. Under some conditions (for example, when $\Sigma$ is an open submanifold of a Riemann surface, where $dx$ is a globally defined meromorphic differential, see~\cite{E}, and we should assume some relation between $y$ and $B$, see~\cite{DOSS} and below), we can represent this set of differentials in terms of the correlators of a CohFT multiplied by some auxiliary differentials. This representation is not canonical, the choice of it is controlled by the action of the group $(\bbC^*)^r$. 

Our goal is to make this action on all ingredients of the formula (that is, the matrix $R$ of a CohFT, its underlying TFT, and the auxiliary differentials) precise.

\subsection{The formula} We fix a point $(C_1,\dots,C_r)\in (\bbC^*)^r$.  We also fix some additional constant $C\in \bbC^*$. All constructions in this Section depend on these choices.

We choose a local coordinate $w_i$ on $U_i$, $i=1,\dots,r$, such that $w_i(p_i)=0$ and
\begin{equation}
x=(C_iw_i)^2+x_i.
\end{equation}
In this case, the underlying TFT is given by 
\begin{align}\label{eq:C-underlyingTFT}
& \eta(e_i,e_j)= \delta_{ij}, \\ \notag
& \alpha^{Top}_{g,n}(e_{i_1}\otimes \cdots \otimes e_{i_n}) = \delta_{i_1\dots i_n} \left(-2C_i^2 C \frac{\dd y}{\dd w_i}(0)\right)^{-2g+2-n}.
\end{align}
In particular, the unit vector is equal to $\sum_{i=1}^r\left(-2C_i^2 C \frac{\dd y}{\dd w_i}(0)\right) e_i$.

The matrix $R(\zeta)$ is given by
\begin{equation}\label{eq:C-matrixR}
-\frac 1\zeta R^{-1}(\zeta)_i^j=\frac{1}{\sqrt{2\pi\zeta}} \int_{-\infty}^\infty \left. \frac{B(w_i,w_j)}{\dd w_i}\right|_{w_i=0} \cdot  e^{-\frac{w_j^2}{2\zeta}}.
\end{equation}
We have to check that the function $y$ satisfies the condition
\begin{equation}\label{eq:condition-y}
\frac{2C_i^2 C}{\sqrt{2\pi\zeta}} \int_{-\infty}^\infty \dd y\cdot e^{-\frac{w_i^2}{2\zeta}} = \sum_{k=1}^r (R^{-1})^i_k \left(2C_k^2 C \frac{\dd y}{\dd w_k}(0)\right)
\end{equation}

Finally, the auxiliary functions $\xi_i\colon \Sigma\to\bbC$ are given by
\begin{equation}\label{eq:C-functionsxi}
\xi_i(x):=\int^x \left.\frac{B(w_i,w)}{\dd w_i}\right|_{w_i=0}
\end{equation}

Using Formulas~\eqref{eq:C-underlyingTFT} and~\eqref{eq:C-matrixR} we define a CohFT, whose classes we denote by $\alpha^{Coh}_{g,n}(e_{i_1}\otimes \cdots \otimes e_{i_n})$. 

\begin{theorem}\label{thm:CohFT} \cite{E,DOSS} The differentials $\cW_{g,n}$ produced by the topological recursion from the input $(\Sigma,x,y,B)$ are equal to
\begin{align}\label{eq:C-fullformula}
 \cW_{g,n}=  C^{2g-2+n}\sum_{\substack{i_1,\dots,i_n \\ d_1,\dots,d_n}}
\int_{\oM_{g,n}}\alpha^{Coh}_{g,n}(e_{i_1}\otimes \cdots \otimes e_{i_n}) 
\\ \notag
\prod_{j=1}^n \psi_j^{d_j} \dd \left(\left( - \frac{1}{w_j}\frac{\dd}{\dd w_j}\right)^{d_j} \xi_{i_j} \right).
\end{align}
In particular, this formula doesn't depend on the choice of $(C_1,\dots,C_r)\in(\bbC^*)^r$ and $C\in\bbC^*$, though all its ingredients do.
\end{theorem}

The proof of this theorem is given by exactly the same argument as in~\cite{E,DOSS}, with a different choice of local coordinates near the points $p_i$, so we omit it here. 

\begin{remark}\label{rem:cohft/1} Let us discuss what happens if the condition~\eqref{eq:condition-y} is not satisfied. Still, under the same conditions a version of Theorem~\ref{thm:CohFT} holds. Namely, we can represent the correlation differentials as
\begin{align}\label{eq:C-fullformula-1}
\cW_{g,n}=  C^{2g-2+n}\sum_{\substack{i_1,\dots,i_n \\ d_1,\dots,d_n}}
\int_{\oM_{g,n}}\alpha^{Coh/1}_{g,n}(e_{i_1}\otimes \cdots \otimes e_{i_n}) 
\\ \notag
\prod_{j=1}^n \psi_j^{d_j} \dd \left(\left( - \frac{1}{w_j}\frac{\dd}{\dd w_j}\right)^{d_j} \xi_{i_j} \right),
\end{align}
where the classes $\alpha^{Coh/1}_{g,n}$ are described, in terms of the graphical formalism recalled in Section~\ref{sec:CohFT/1}, via the same TFT and $R$-matrix as
$\alpha^{Coh}_{g,n}$ in Theorem~\ref{thm:CohFT}, but instead of the dilaton leaves, we decorate each vertex labeled by $i$ (that is, the one that is decorated by $\alpha_{g,n}^{Top}(e_i\otimes\cdots\otimes e_i)$) with the $\kappa$-class
\begin{equation}
\exp\left(\sum_{k=1}^\infty T_{i,k}\kappa_k \right),
\end{equation}
where the constants $T_{i,k}$ are given by 
\begin{equation}
\frac{\dd y}{\dd w_i}(0)\exp\left(\sum_{k=1}^\infty T_{i,k}(-\zeta)^k \right)=
\frac{1}{\sqrt{2\pi\zeta}} \int_{-\infty}^\infty \dd y\cdot e^{-\frac{w_i^2}{2\zeta}} .
\end{equation}
This is a direct corollary of~\cite[Theorem 3.2]{Eyn11}, see also \cite[Lemma 3.5]{DOSS}.
\end{remark}

\section{Computations with the spectral curve}

Consider the following initial data on the spectral curve $\Sigma=\CP1$ with a global coordinate $z$:
\begin{align}\label{eq:spectral-curve-data}
	& x(z)=-z^r+\log z; \\ \notag
	& y(z)=z^s; \\ \notag
	& B(z,z') = \frac{\dd z\, \dd z'}{(z-z')^2}.
\end{align}

In this section we compute all ingredients of the Formula~\eqref{eq:C-fullformula} for this initial data with a special choice of the torus point.  In particular, for $1\leq s \leq r$ we prove that the correlation differentials are controlled by a CohFT, and the corresponding CohFT coincides with the one given by Chiodo classes~\eqref{eq:DefinitionChiodoCohFT} considered in the normalized canonical frame. 

\subsection{Local expansions}
As it was computed in~\cite{SSZ}, we can associate with this curve the following local data.

The critical points are 
\begin{equation}
p_i := r^{-1/r} \J^i, \qquad i=0,\dots, r-1,
\end{equation}
and the critical values of the function $x$ at these points are 
\begin{equation}
x_i:= x(p_i) = -\frac{1}{r} + \frac{2\pi i \I}{r} - \frac{\log r}{r}, \qquad i=0,\dots, r-1.
\end{equation}
If we choose a local coordinate $w_i$ near the point $p_i$ such that $w_i(p_i)=0$ and $-w_i^2/2r+x_i=x$, $i=0,1,\dots,r-1$, then there are two possible choices for the expansion of the function $z$ in $w_i$. We fix it to be
\begin{equation}
z(w_i)=r^{-1/r} \J^i + \left(r^{-1 -\frac{1}{r}} \J^i\right) w_i +O(w_i^2), 
\end{equation}
With this choice we also fix the expansion of $y=z^s$, namely,
\begin{equation}
y(w_i)=r^{-s/r}\J^{si}+\left(s r^{-1 -\frac{s}{r}} \J^{is}\right) w_i +O(w_i^2). 
\end{equation}

\begin{lemma}
	We have: 
	\begin{equation}\label{eq:y-laplace}
		\frac{1}{\sqrt{2\pi\zeta}} \int_{-\infty}^\infty \dd y(w_i)\cdot e^{-\frac{w_i^2}{2\zeta}}
		\sim
		\left(s r^{-1 -\frac{s}{r}} \J^{is}\right) \exp\left( -\sum_{m=1}^\infty \frac{B_{m+1}\left(\frac sr \right)}{m(m+1) }
		(-\zeta)^m \right). 
	\end{equation}
\end{lemma}

\begin{proof}
	This Lemma is analogous to~\cite[Lemma 4.3]{SSZ}. Indeed, we introduce a new coordinate $t=rz^r$. In this coordinate we have:
	\begin{align}
	z & = t^{\frac 1r} r^{-\frac 1r} \J^i; \\
	-x_i-z^r+\log z & = \frac{1}{r}-\frac{t}{r} +\frac{\log t}{r}; \\
	\dd z & = t^{\frac{1-r}{r}} r^{-1-\frac{1}{r}} \J^i \dd t.
	\end{align}
We can then make a change of variables and use the standard asymptotic expansion of the gamma function, cf. the proof of Lemma 4.3 in~\cite{SSZ}:	 
	\begin{align}
		& \frac{\sqrt{-2 r}}{\sqrt{2\pi\zeta}} \int \dd y\cdot e^{2r\cdot \frac{(x-x_i)}{2\zeta}}  = 
		\frac{s r^{-\frac 12-\frac sr} \J^{si} e^{\frac 1\zeta}}{\sqrt{-\pi\zeta}} \int \dd t \cdot t^{\frac{s-r}{r}+\frac 1\zeta}
		e^{-\frac t\zeta} \\ \notag
		& \sim
		\left( s\sqrt{-2} r^{-\frac 12 -\frac sr} \J^{si}\right)  \exp\left( -\sum_{m=1}^\infty \frac{B_{m+1}\left(\frac sr \right)}{m(m+1) }
		(-\zeta)^m \right).
	\end{align}
\end{proof}

\begin{lemma}
	We have: 
	\begin{align}
	& \frac{1}{\sqrt{2\pi\zeta}} \int_{-\infty}^\infty \left. \frac{B(w_i,w_j)}{\dd w_i} \right|_{w_i=0}  \cdot e^{-\frac{w_j^2}{2\zeta}} \\ \notag
	& \sim \sum_{c=0}^{r-1}\frac{\J^{cj-ci}}{r}
	\frac{ \exp\left( -\sum_{m=1}^\infty \frac{B_{m+1}\left(\frac cr \right)}{m(m+1) }
	(-\zeta)^m \right)} {(-\zeta)}. 
	\end{align}
\end{lemma}

\begin{proof}
	This Lemma is just a refined version of Lemma 4.4 in~\cite{SSZ}, so the proof is exactly the same as there.
\end{proof}	

Note that this Lemma means that we have to consider the Givental group action defined by the matrix $R(\zeta)$, where 
\begin{equation}
R^{-1}(\zeta)^j_i := \sum_{c=0}^{r-1}\frac{\J^{cj-ci}}{r}
{ \exp\left( -\sum_{m=1}^\infty \frac{B_{m+1}\left(\frac cr \right)}{m(m+1) }
	(-\zeta)^m \right)} .
\end{equation}
We choose the constants $C_1=\cdots=C_r:=1/\sqrt{-2r}$ and $C:=r^{1+s/r}/s$. In particular, with this choice the structure constants of the underlying TFT are given by
\begin{equation}\label{eq:determineTFT}
-2C_i^2 C \frac{\dd y}{\dd w_i}(0) = \frac{\J^{is}}{r}
\end{equation}

\begin{lemma} \label{Lem:DilatonLeaves}
For $1\leq s\leq r$ the	condition~\eqref{eq:condition-y} is satisfied.
\end{lemma}

\begin{proof} This is a direct computation. We have:
\begin{align}
& 
\frac{2C_i^2 C}{\sqrt{2\pi\zeta}} \int_{-\infty}^\infty \dd y\cdot e^{-\frac{w_i^2}{2\zeta}} 
= -\frac{\J^{is}}{r}\exp\left( -\sum_{m=1}^\infty \frac{B_{m+1}\left(\frac sr \right)}{m(m+1) } (-\zeta)^m \right)
\\ \notag &
=\sum_{k=1}^r \sum_{c=0}^{r-1} \frac{\J^{ci-ck}}{r} \exp\left( -\sum_{m=1}^\infty \frac{B_{m+1}\left(\frac cr \right)}{m(m+1) } (-\zeta)^m \right) \left( - \frac{\J^{ks}}r \right) 
\\ \notag &
 = \sum_{k=1}^r (R^{-1})^i_k \left(2C_k^2 C \frac{\dd y}{\dd w_k}(0)\right)
\end{align}
The second equality is true for $0 \leq s \leq r-1$, and also for $s = r$,
since $B_{m + 1}(1) = B_{m+1}(0)$ for $m \geq 1$.
\end{proof}

This Lemma implies that we indeed have correlators of a cohomological field theory inside Formula~\eqref{eq:C-fullformula} in this case. 

Finally, Definition~\eqref{eq:C-functionsxi} implies that 
\begin{equation}
\xi_i = \frac{r^{-1-\frac 1r}\J^i}{r^{-\frac 1r}\J^i -z},
\end{equation}
and it is easy to see that 
\begin{equation}\label{eq:derivative-psi-class}
-\frac {1}{w} \frac{\dd}{\dd w} = \frac 1r\frac{\dd}{\dd x}.
\end{equation}
This completes the description of all the ingredient of the Formula~\eqref{eq:C-fullformula} for the correlation differentials $\cW_{g,n}$.


\subsection{Correlation differentials in flat basis}

In the previous section we described all ingredients of the formula for the correlation differentials~\eqref{eq:C-fullformula} for the case of the spectral curve data~\eqref{eq:spectral-curve-data}. In particular, for $1\leq s\leq r$ we proved that there are the correlators of a CohFT inside this formula, otherwise we have a CohFT/1. Our goal now is to show that the cohomological field theories obtained in the previous Section is the one given by the same formulas as in Lemmas~\ref{lem:givental-cohft} and~\ref{lem:givental-CohFT-1}. In order to do that we apply a linear change of variables to the basis $e_0,\dots,e_{r-1}$ used in the previous Section. 

We use the change of basis from $e_0,\dots,e_{r-1}$ to $v_1,\dots,v_r$ given by the formula 
\begin{equation}
e_i=\sum_{a=1}^r {\J^{-ai}}v_a; \qquad v_a=\sum_{i=0}^{r-1} \frac{\J^{ai}}r e_i   
\end{equation}

\begin{lemma} \label{lem:flat-basis}
	 In the basis $v_1,\dots,v_r$ we have:
\begin{itemize}
\item[$\bullet$] The underlying TFT $\alpha_{g,n}^{Top}$~\eqref{eq:C-underlyingTFT} with the choice of constants given by Equation~\eqref{eq:determineTFT} is given by
\begin{align} \label{eq:TFT-flat-basis}
& \eta(v_a,v_b)=\frac{1}{r}\delta_{a+b\mod r}; \\ \notag 
& \omega_{0,3}(v_{a}\otimes v_b \otimes v_{c}) =\frac 1r \delta_{a+b+c-s \mod r}
\\ \notag
& \omega_{g,n}(v_{a_1}\otimes \cdots \otimes v_{a_n}) =r^{2g-1} \delta_{a_1+\cdots+a_n-s(2g-2+n) \mod r}
\end{align}
\item[$\bullet$]
The $R$-matrix is given by 
\begin{equation}\label{eq:Rmatrix-flat-basis}
R(\zeta)= \exp\left(\sum_{m=1}^{\infty} \frac{\mathrm{diag}_{a=1}^r B_{m+1}\left(\frac ar\right)}{m(m+1)}(-\zeta)^m\right)
\end{equation}
\item[$\bullet$]
The auxiliary functions $\xi_a$ are given by
\begin{equation}\label{eq:xi-flat-basis}
\xi_a=r^{\frac{r-a}r} \sum_{p=0}^\infty \frac{(pr+r-a)^p}{p!}e^{(pr+r-a)x}.
\end{equation}
\end{itemize}
\end{lemma}

\begin{proof}
The computation of the underlying TFT is fairly simple:
\begin{align}
& \eta(v_a,v_b) 
= \sum_{i,j=0}^{r-1} \frac{\J^{ai+bj}}{r^2} \eta(e_i,e_j)
= \sum_{i=0}^{r-1} \frac{\J^{(a+b)i}}{r^2} 
= \frac{1}{r}\delta_{a+b\mod r},
\\ \notag
& \omega_{0,3}(v_a\otimes v_b\otimes v_c) = 
\sum_{i=0}^{r-1} \frac{\J^{ai+bi+ci}}{r^3} \omega_{0,3}(e_i\otimes e_i\otimes e_i) 
\\ \notag
& = \sum_{i=0}^{r-1} \frac{\J^{ai+bi+ci-si}}{r^2}
=\frac{1}{2} \delta_{a+b+c-s\mod r},
\end{align}
and the other correlators of the underlying TFT are determined uniquely. 

The change of basis for the matrix $R^{-1}$ reads:
\begin{align}
R^{-1}(\zeta)^b_a & = 
\sum_{i,j=0}^{r-1} \frac{\J^{-jb+ia}}{r} \sum_{c=0}^{r-1}\frac{\J^{cj-ci}}{r}
{ \exp\left( -\sum_{m=1}^\infty \frac{B_{m+1}\left(\frac cr \right)}{m(m+1) }
	(-\zeta)^m \right)}
\\ \notag
& = 
{ \exp\left( -\sum_{m=1}^\infty \frac{B_{m+1}\left(\frac cr \right)}{m(m+1) }
	(-\zeta)^m \right)}\cdot 
\delta_{c-b\mod r}\cdot \delta_{c-a\mod r} 
\\ \notag
& =
{ \exp\left( -\sum_{m=1}^\infty \frac{B_{m+1}\left(\frac ar \right)}{m(m+1) }
	(-\zeta)^m \right)}\cdot 
\delta_{a-b},
\end{align}
which implies Equation~\eqref{eq:Rmatrix-flat-basis}.

Finally, Equation~\eqref{eq:xi-flat-basis} follows from Lemma 4.6 in~\cite{SSZ}.
\end{proof}

\begin{remark}\label{rem:coincides-with-Chiodo}
Observe that Equations~\eqref{eq:TFT-flat-basis} and~\eqref{eq:Rmatrix-flat-basis} and Lemma~\ref{Lem:DilatonLeaves} imply that for $s\leq r$ the cohomological field theory that we have in the flat basis coincides with the one given in Lemma~\ref{lem:givental-cohft}.  For $s> r$, where Lemma~\ref{Lem:DilatonLeaves} does not apply, we have obtained the topological field theory and the $R$-matrix as in Lemma~\ref{lem:givental-CohFT-1}, but we still have to compare the power series that determines the $\kappa$-legs. 
\end{remark}

 Lemma~\ref{lem:flat-basis} allows us to rewrite  formula~\eqref{eq:C-fullformula} for the correlation differentials of the spectral curve data~\eqref{eq:spectral-curve-data} in the following way.

\begin{theorem}\label{TH1}
The correlation differentials of the spectral curve~\eqref{eq:spectral-curve-data} are equal to
	\begin{align}\label{eq:correlation-forms}
	\cW_{g,n}= &  \sum_{\mu_1,\dots,\mu_n=1}^\infty
	\dd_1\otimes \cdots \otimes \dd_n\ e^{\sum_{j=1}^n \mu_j x_j}
	\\ \notag &
	\times
	\int_{\oM_{g,n}} \frac{\C_{g,n} \left (r,s; r-r \left < \frac{\mu_1}{r} \right >,
          \dots,
          r - r \left < \frac{\mu_n}{r} \right > \right )
		}{\prod_{j=1}^n (1-\frac{\mu_i}{r}\psi_i)}
		\\ \notag &
		\times
		\prod_{j=1}^{n} \frac{\left(\frac{\mu_j}{r}\right)^{\floor*{\frac{\mu_j}r}}}{\floor*{\frac{\mu_j}r}!}
		\times
		\frac{r^{2g-2+n+\frac{(2g-2+n)s+\sum_{j=1}^n \mu_j}r}}{s^{2g-2+n}},
	\end{align}
where $\frac{\mu}r = \lfloor \frac{\mu}r \rfloor + \langle \frac{\mu}r \rangle$ is the decomposition into the integer and the fractional parts.
\end{theorem}

\begin{proof} First, consider the case $s\leq r$.
	Using Equation~\eqref{eq:C-fullformula}, together with Lemma~\ref{lem:flat-basis}, Remark~\ref{rem:coincides-with-Chiodo},  Equation~\eqref{eq:derivative-psi-class} and $C=r^{1+s/r}/s$, we have:
	\begin{align}
	& \cW_{g,n}(x_1,\dots,x_n)  
	\\ \notag &
	=\sum_{\substack{d_1,\dots,d_n\geq 0 \\ 1\leq a_1,\dots,a_n\leq r}}
 \frac{r^{2g-2+n+\frac{(2g-2+n)s}r}}{s^{2g-2+n}}
	\int_{\oM_{g,n}} \C_{g,n}(r,s;a_1,\dots,a_n)
	\\ \notag &
	 \times \prod_{j=1}^n \psi_j^{d_j} 
	r^{-d_j} r^{\frac{r-a_j}r} \dd \left[\left(\frac{\dd}{\dd x_j}\right)^{d_j}
	\sum_{p=0}^\infty
	 \frac{(pr+r-a_j)^p}{p!}e^{(pr+r-a_j)x_j}\right]
	 \\ \notag &
	 =\dd_1\otimes\cdots\otimes \dd_n
	 \sum_{\substack{d_1,\dots,d_n\geq 0 \\ 1\leq a_1,\dots,a_n\leq r}}
	 \int_{\oM_{g,n}} \C_{g,n}(r,s;a_1,\dots,a_n) \prod_{j=1}^n \psi_j^{d_j}
	 \\ \notag &
	 \times \frac{r^{2g-2+2n-\sum_{j=1}^n d_j+\frac{(2g-2+n)s-\sum_{j=1}^n a_j}r}}{s^{2g-2+n}}
	 \\ \notag &
	 \times \prod_{j=1}^n 
	 	 \sum_{p=0}^\infty
	 \frac{(pr+r-a_j)^{p+d_j}}{p!}e^{(pr+r-a_j)x_j}.
	\end{align}
	Equation~\eqref{eq:correlation-forms} is just a way to rewrite the last formula using a summation over the parameter $\mu_i = p_ir+r-a_i$ instead of a double summation over $p_i$ and $a_i$.
	
	In the case $s>r$, we should compute separately the $\kappa$-classes. In this case, Remark~\ref{rem:cohft/1} and Equation~\eqref{eq:y-laplace}  imply that the $\kappa$-class attached to the vertex of index $i$ (in the basis $e_0,\dots,e_{r-1}$) is equal to 
	$\exp\left(\sum_{m=1}^\infty (-1)^m \frac{B_{m+1}(\frac sr)}{m(m+1)} \kappa_m \right)$. Since it doesn't depend on~$i$, it remains the same in the basis $v_1,\dots,v_r$, where it coincides with the one given by Lemma~\ref{lem:givental-CohFT-1}.
\end{proof}

\begin{remark}
	Note that in the case $s=1$ we reproduce Theorem 1.7 in~\cite{SSZ}.
\end{remark}

\section{Johnson-Pandharipande-Tseng formula and topological recursion}

In this Section we consider a special case of the correspondence between the Chiodo formulas and the spectral curve topological recursion. We assume that $s=r$. In this case, the correlation differentials of this spectral curve are known to give the so-called $r$-orbifold Hurwitz numbers in some expansion. 

An $r$-orbifold Hurwitz number $h_{g;\vec{\mu}}$  is just a double Hurwitz number that enumerates ramified coverings of the sphere by a genus $g$ surface, where one special fiber is arbitrary (given by the partition $\vec\mu$ of length $n$) and one has ramification indices $(r,r,\dots,r)$. Therefore, the degree of the covering $\sum_{i=1}^{n}\mu_i$ is divisible by $r$ and  there are $b=2g-2+n+\sum_{i=1}^{n}\mu_i/r$ simple critical points.

The $r$-orbifold Hurwitz numbers are also known to satisfy the John\-son-Pandharipande-Tseng (JPT) formula that expresses them in terms of the intersection theory of the moduli space of curves. The main goal of this Section is to show that the JPT formula is equivalent to the topological recursion for $r$-orbifold Hurwitz numbers. In particular, this gives a new proof of the topological recursion for $r$-orbifold Hurwitz numbers.

\subsection{The JPT formula}
The formula of Johnson, Pandharipande and Tseng is presented in~\cite{JoPaTs} for a general abelian group ${G}$, its particular finite representation $U$ and a vector of monodromies $\gamma$. Here we consider only the case of
${G} = \mathbb{Z}/r\mathbb{Z}$, the representation $U$ sends $1 \in \mathbb{Z}/r\mathbb{Z}$ to $e^{\frac{2 \pi i}{r}}$, and $\gamma$ is empty. In this case the JPT formula reads

\begin{align} \label{eq:particular-jpt}
 \frac{ h_{g;\vec{\mu}} } {b!} = 
  r^{1-g+\sum \left < \frac{\mu_i}{r} \right >}
  \prod_{i=1}^{n}\frac{\mu_i^{\floor*{\frac{\mu_i}{r}}}}{\floor{\frac{\mu_i}{r}}!}
  \int_{\overline{\mathcal{M}}_{g,n}}
  \frac{\epsilon_*\sum_{i\geq 0} (-r)^i \lambda_i}{\prod_{j=1}^{n} (1 - \mu_j {\psi}_j)},
\end{align}
where the class $\epsilon_*\sum_{i\geq 0} (-r)^i \lambda_i$ is described in detail below.

\subsection{Two descriptions of $r$th roots}\label{sec:geometry}

Let $G = \Z/r\Z$ be the abelian group of $r$th roots of unity. The space $\oM_{g;a_1, \dots, a_n}(\cB G)$ is the space of stable maps to the stack $\cB G$ with monodromies $a_i \in \{0, \dots, r-1 \}$ at the marked points. This space, and the natural cohomology classes on it, 
can be constructed in several ways, see, for instance,~\cite{AbrVis,CheRua}. Johnson, Pandharipande, and Tseng~\cite{JoPaTs} use the description via admissible covers. Chiodo~\cite{Chiodo} uses the moduli space of $r$th roots of the line bundle $\cO(-\sum a_i x_i)$. In our work we apply Chiodo's formulas to a result of Johnson, Pandharipande, and Tseng, so we recall and briefly explain the equivalence between the two approaches.

\subsubsection{The $r$-stable curves.}
An {\em $r$-stable curve} is an orbifold stable curve whose only nontrivial orbifold structure appears at the nodes and at the markings. The neighborhood of a marking is isomorphic to $\Delta/G$, where an $r$th root of unity $\rho \in G$ acts on the disc $\Delta$ by $z \mapsto \rho z$. The neighborhood of a node in a family of $r$-stable curves is isomorphic to $(\Delta \times \Delta)/G$, where $\rho \in G$ acts by $(z,w) \mapsto (\rho z, \rho^{-1} w)$.

The moduli space of $r$-stable curves has the same coarse space as $\oM_{g,n}$, but an extra factor of $G$ appears in the stabilizer for every node of the curve.

\subsubsection{Line bundles over $r$-stable curves.} 
A line bundle $L$ over an $r$-stable curve has a particular structure at the neighborhoods of markings and nodes. At a marking it can be given by the chart $\Delta \times \bbC$ with the action of an element $\rho \in G$ given by $(z, s) \mapsto (\rho z, \rho^as)$. Thus the number $a \in \{0, \dots, r-1 \}$ describes the local structure of $L$ at a marking. At a node $L$ can be given a by a chart $(\Delta \times \Delta) \times \bbC$ with the action of an element $\rho \in G$ given by $(z, w, s) \mapsto (\rho z, \rho^{-1} w, \rho^a s)$. Note, however, that the number $a$ is replaced with $-a \pmod r$ if we exchange $z$ and~$w$. Thus the local structure of $L$ at node is described by assigning  to the branches of the node two numbers $a', a'' \in \{ 0, \dots, r-1 \}$ such that $a'+a''=0 \mod r$.

\subsubsection{Roots of $\cO$.} In~\cite{Chiodo} an element of $\oM_{g;a_1, \dots, a_n}(\cB G)$ is an $r$-stable curve~$\cC$ with an orbifold line bundle~$L \to \cC$ endowed with an identification $L^{\otimes r} \ident \cO$. The integers $a_i \in \{0, \dots, r-1 \}$ prescribe the structure of $L$ at the markings.

\subsubsection{From $r$-th roots to $G$-bundles.}
To make the connection with the description of $\oM_{g;a_1, \dots, a_n}(\cB G)$ in~\cite{JoPaTs} we look at the multi-section of $L$ that maps to the section $1$ of $\cO$ when raised to the power~$r$. This multi-section is a principal $G$-bundle $\pi: D \to C$ ramified over the markings and the nodes. At a marking with label~$a$ the $G$-bundle has the monodromy given by adding~$a$ in $\Z/r\Z$. This can be seen from the $G$-action $(z, s) \mapsto (\rho z, \rho^as)$. If we choose $\rho = e^{2 \pi i/r}$, a path from $z$ to $\rho z$ in the chart corresponds to a loop around the marking in the stable curve and its lifting leads from $s$ to $\rho^as$ in the fiber of~$L$. 

Similarly, at the node the $G$-bundle has monodromies $a'$ and $a''$ at the two branches, satisfying $a' + a'' = 0 \mod r$. 

Note that, because $D$ is formed by a multi-section of~$L$, the pull-back of $L$ to $D$ has a tautological section. We will denote this section by $\phi_0$.

\subsubsection{From $G$-bundles to $r$-th roots.} In~\cite{JoPaTs} an element of $\oM_{g;a_1, \dots, a_n}(\cB G)$ is $G$-cover $\pi:D \to C$ ramified over the markings and the nodes and satisfying the ``kissing condition'': the monodromies of the $G$-action over two branches of a node are opposite modulo~$r$. The integers $a_i \in \{0, \dots, r-1 \}$ prescribe the monodromies at the markings. Suppose we are given a principal $G$-bundle $\pi:D \to C$ like that. Using this data it is easy to construct a line bundle $L$ over the $r$-stable curve $\cC$ corresponding to~$C$. Over any contractible open set $U \subset C$ that does not contain markings and nodes we create a chart $U \times \bbC$ and identify the $r$-roots of unity in $\bbC$ with the sheets of the $G$-bundle in an arbitrary way that preserves the $G$-action. At the markings we create the orbi-chart $\Delta \times \bbC$ endowed with the $G$-action $(z, s) \mapsto (\rho z, \rho^as)$ as above and also identify the $r$-th roots of unity with the sheets of the bundle. The transition maps between the charts are obtained from the matching of the sheets over different charts (every transition map is the multiplication by a locally constant $r$-th root of unity). 

\subsubsection{Sections of $L$ and of $K \otimes L^*$.} Let $\phi$ be a section of~$L$ over an open set $U \subset \cC$. Then $\pi^* \phi/ \phi_0$ is a holomorphic function on $\pi^{-1}(U) \subset D$. Moreover, the $G$-action on this function has the form $f(\rho z) = \rho^{-1} f(z)$. A global section of $L$ gives rise to a global holomorphic function on $D$ satisfying the above transformation rule. It follows that $L$ has no global sections over $\cC$, with the exception of the case where all $a_i$'s vanish, $L$ is the trivial line bundle and $D = C \times G$. 

 Similarly, let $\phi$ be a section of $K \otimes L^*$ on an open set $U \subset \cC$. Then $\alpha = \pi^* \phi \cdot \phi_0$ is a section of the canonical line bundle $K_D$ over $\pi^{-1}(U)$. Moreover, the $G$-action on this function has the form $\alpha(\rho z) =\rho \alpha(z)$. In particular, the space of global sections of $K \otimes L^*$ coincides with the space of holomorphic differentials on $D$ satisfying the transformation rule $\alpha(\rho z) =\rho \alpha(z)$.
 
\subsubsection{Two ways of writing $R^*p_*L$.}
Chiodo's formula expresses the Chern character of $R^*p_*L$, where $p: \oC_{g;a_1, \dots, a_n}(\cB G) \to \oM_{g;a_1, \dots, a_n}(\cB G)$ is the universal curve. Using this formula one can also easily express the total Chern class of $-R^*p_*L$.

According to our remarks above, if there is at least one positive $a_i$ then $R^0p_*L =0$. In that case $R^1p_*L$ is a vector bundle, and we have
$c(-R^*p_*L) = c(R^1p_*L)$.

If all the $a_i$'s vanish, the space $\oM_{g;a_1, \dots, a_n}(\cB G)$ has a special connected component on which the line bundle $L$ is trivial. Over this component $R^0p_*L = \bbC$. On the other connected components we have, as before, $R^0p_*L = 0$. Therefore the total Chern class of $R^0p_*L$ is equal to 1 and we have, once again,
$c(-R^*p_*L) = c(R^1p_*L)$.

Johnson, Pandharipande, and Tseng use the Chern classes $\lambda_i$ of the vector bundle of equivariant sections of $K_D$. Our analysis above shows that this vector bundle is the dual of $R^1p_*L$. In other words, we have
\begin{equation} \label{eq:chiodo-orbifold-specialize}
c(-R^*p_*L) = \sum (-1)^i \lambda_i,
\end{equation}
which is the equality that we use in our computations. 

\begin{remark}
In the Johnson-Pandharipande-Tseng formula the monodromies at the markings are given by the remainders modulo $r$ of $-\mu_i$, that is, minus the parts  of the ramification profile. Thus if we denote by $a_i = \mu_i \mod r$, we will use Chiodo's formula with remainders $r-a_1, \dots, r-a_n$ at the markings. If an $a_i$ is equal to 0, we can plug either 0 or $r$ in Chiodo's formula. Indeed, we have $B_k(0) = B_k(1)$ for any $k>1$, thus replacing $0$ by $r$ will only affect the Chern character of degree~0, that is not used in the expression for the total Chern class. 

In particular,  in Equation~\eqref{eq:particular-jpt} we use the push-forward of $\sum (-1)^i \lambda_i$ to $\overline{\mathcal{M}}_{g,n}$, for monodromies equal to minus the remainders of $\mu_1,\dots,\mu_n$. This class coincides with $\C_{g,n}(r,s;r-a_1,\dots,r-a_n)$ defined by Equation~\eqref{eq:DefinitionChiodoCohFT}.
\end{remark}

\subsection{The equivalence}
Now we are armed to prove the following
\begin{theorem}\label{TH2} The expansion of the correlation differentials of the spectral curve~\eqref{eq:spectral-curve-data} for $s = r$ is given by
\begin{align} \label{eq:orbifold-top-rec}
  \cW_{g,n} = &  \sum_{\mu_1,\dots,\mu_n=1}^\infty
	\dd_1\otimes \cdots \otimes \dd_n\ e^{\sum_{j=1}^n \mu_j x_j} \frac{ h_{g;\vec{\mu}} } {b!},
\end{align}
if and only if the numbers $h_{g;\vec{\mu}}$ are given by the Johnson-Pandharipande-Tseng formula \eqref{eq:particular-jpt}.
\end{theorem}
\begin{proof}
The proof is indeed very simple. First, Equation~\eqref{eq:chiodo-orbifold-specialize}
allows us to replace Chiodo class in \eqref{eq:correlation-forms} with the push-forward of the linear combination of $\lambda$-classes.
Then we notice the following rescaling of the integral
\begin{equation}\label{eq:r-rescaling}
  \int_{\oM_{g,n}} \frac{\pi_*\sum_{i\geq 0} (-r)^i \lambda_i
  }{\prod_{j=1}^n (1-\mu_i \psi_i)}
  = r^{3g-3+n}
  \int_{\oM_{g,n}} \frac{\pi_*\sum_{i\geq 0} (-1)^i \lambda_i
  }{\prod_{j=1}^n (1-\frac{\mu_i}{r}\psi_i)}.
\end{equation}
The equivalence then follows from comparison of coefficients in front of particular
$\dd_1\otimes \cdots \otimes \dd_n\ e^{\sum_{j=1}^n \mu_j x_j}$ in \eqref{eq:orbifold-top-rec}
and \eqref{eq:correlation-forms}, which is obvious, modulo the following simple computation of the powers of $r$. For $s=r$,
$$\prod_{j=1}^{n} \frac{\left(\frac{\mu_j}{r}\right)^{\floor*{\frac{\mu_j}r}}}{\floor*{\frac{\mu_j}r}!}
\frac{r^{2g-2+n+\frac{(2g-2+n)s+\sum_{j=1}^n \mu_j}r}}{s^{2g-2+n}}
=\prod_{j=1}^{n} \frac{\mu_j^{\floor*{\frac{\mu_j}r}}}{\floor*{\frac{\mu_j}r}!}
r^{2g-2+n+\sum_{j=1}^n \left < \frac{\mu_i}{r} \right >}
$$
is the coefficient in Equation~\eqref{eq:correlation-forms}. This is equal to 
$$
r^{3g-3+n}   r^{1-g+\sum \left < \frac{\mu_i}{r} \right >}
\prod_{i=1}^{n}\frac{\mu_i^{\floor*{\frac{\mu_i}{r}}}}{\floor{\frac{\mu_i}{r}}!},
$$
which is the coefficient of~\eqref{eq:particular-jpt} after rescaling~\eqref{eq:r-rescaling}.
\end{proof}

\end{document}